\documentclass{llncs}
\usepackage{fullpage}
\usepackage{epsfig} 
\usepackage{times} 
\usepackage{gensymb}
\usepackage[english]{babel} 
\usepackage{graphicx} 
\usepackage{amssymb}
\usepackage{multirow}
\usepackage{xspace}
\usepackage{xcolor}
\usepackage{subfigure}
\usepackage{paralist}
\usepackage{amssymb}

\newcommand{\remove}[1]{}

\renewcommand{\int}{int} \newtheorem{claimx}{Claim}
 
\renewenvironment{proof} {{\em
    Proof:}}{\hspace*{\fill}$\Box$\par\vspace{2mm}}

\newcommand{\NP}{$\mathcal{NP}$\xspace}

\renewcommand{\NP}{$\mathcal{NP}$}
\newcommand{\NPC}{\mbox{\NP-complete}\xspace}

\newcommand{\NPCN}{\mbox{\NP-completeness}\xspace}

\newcommand{\NPHN}{\mbox{\NP-hardness}\xspace}

\definecolor{blue}{rgb}{0.274,0.392,0.666}
\definecolor{red}{rgb}{0.627,0.117,0.156}
\definecolor{green}{rgb}{0,0.588,0.509}

\newcommand{\red}[1]{{\color{red}{#1\xspace}}}
\newcommand{\blue}[1]{{\color{blue}{#1\xspace}}}

\newcommand{\clinstance}[1]{$(V^{#1},E^{#1},\gamma^{#1},T^{#1})$\xspace}
\newcommand{\clp}{{\sc Clustered-Level Planarity}\xspace}
\newcommand{\clpshort}{{\sc CL-Planarity}\xspace}
\newcommand{\tlinstance}[1]{$(V^{#1},E^{#1},\gamma^{#1},\mathcal{T}^{#1})$\xspace}

\newcommand{\tlp}{{\sc {\em $T$}-Level Planarity}\xspace}
\newcommand{\betp}{{\sc Betweenness}\xspace}
\newcommand{\Gint}{$G_{\cap}$\xspace}
\newcommand{\Gr}[1]{$\red{G^{#1}_1}$\xspace}
\newcommand{\Gb}[1]{$\blue{G^{#1}_2}$\xspace}
\newcommand{\sefeinstance}[1]{$\langle\red{G^{#1}_1},\blue{G^{#1}_2}\rangle$\xspace}
\newcommand{\sefe}{SEFE-$2$\xspace} 

\title{On the Complexity of Clustered-Level Planarity and $T$-Level Planarity
\thanks{
Research was supported in part by the Italian
  Ministry of Education, University, and Research (MIUR) under PRIN
  2012C4E3KT national research project ``AMANDA -- Algorithmics for
  MAssive and Networked DAta'' and by ESF project 10-EuroGIGA-OP-003 GraDR. Fabrizio Frati was partially supported by the Australian Research 
  Council (grant DE140100708)'.}
}

\date{}

\newcommand{\rome}{$^\dag$} \newcommand{\sidney}{$^{\diamond}$}
\author{Patrizio Angelini\rome, Giordano {Da Lozzo\rome}, Giuseppe {Di Battista\rome}, \\
    Fabrizio Frati\sidney, Vincenzo Roselli\rome \institute{
    \rome~Department of Engineering, Roma Tre University, Italy\\
    \email{\{angelini,dalozzo,gdb,roselli\}@dia.uniroma3.it}\\
    \sidney~School of Information Technologies, The University of Sydney, Australia\\
    \email{fabrizio.frati@sydney.edu.au} }}

\begin{document}

\maketitle

\begin{abstract}
  In this paper we study two problems related to the drawing of level
  graphs, that is, \tlp and \clp. We show that both problems are \NPC in the general case and that they become polynomial-time solvable when restricted to proper instances.
\end{abstract}

\section{Introduction and Overview} \label{se:introduction}
A level graph is {\em proper} if any of its edges spans just two consecutive levels. Several papers about constructing level drawings of level graphs assume that the input graph is proper. Otherwise, they suggest to make it proper by ``simply adding dummy vertices'' along the edges spanning more than two levels. In this paper we show that this apparently innocent augmentation has dramatic consequences if, instead of constructing just a level drawing, we are also interested in representing additional constraints, like clustering of vertices or consecutivity constraints on the ordering of vertices on levels.

A \emph{level graph} $G =(V,E,\gamma)$ is a graph with a function $\gamma: V \rightarrow \{1,2,...,k\}$, with $1 \leq k \leq |V|$ such that $\gamma(u) \neq \gamma(v)$ for each edge $(u,v) \in E$. The set $V_i = \{v | \gamma(v)=i\}$ is the $i$-th \emph{level} of $G$.
A level graph $G = (V,E,\gamma)$ is \emph{proper} if for every edge $(u,v) \in E$, it holds $\gamma(u) = \gamma(v) \pm 1$.
A \emph{level planar drawing} of $(V,E,\gamma)$ maps each vertex $v$ of each level $V_i$ to a point on line $y = i$, denoted by $L_i$, and each edge to a $y$-monotone curve between its endpoints so that no two edges intersect.
A level graph is \emph{level planar} if it admits a level planar drawing.
A linear-time algorithm for testing level planarity was presented by J{\"u}nger and Leipert in~\cite{jl-lpelt-02}.

A \emph{clustered-level graph} (\emph{cl-graph}) $(V,E,\gamma,T)$ is a level graph $(V,E,\gamma)$ equipped with a \emph{cluster hierarchy} $T$, that is, a rooted tree where each leaf is an element of $V$ and each internal node $\mu$, called \emph{cluster}, represents the subset $V_\mu$ of $V$ composed of the leaves of the subtree of $T$ rooted at $\mu$.
A \emph{clustered-level planar drawing} (\emph{cl-planar} drawing) of $(V,E,\gamma,T)$ is a level planar drawing of level graph $(V,E,\gamma)$ such that: (1) each cluster $\mu$ is represented by a simple region enclosing all and only the vertices in $V_\mu$; (2) no edge intersects the boundary of a cluster more than once; (3) no two cluster boundaries intersect each other; and (4) the intersection of $L_i$ with any cluster $\mu$ is a straight-line segment, that is, the vertices of $V_i$ that belong to $\mu$ are consecutive along $L_i$.
A cl-graph is \emph{clustered-level planar} (\emph{cl-planar}) if it admits a cl-planar drawing.
\clp (\clpshort) is the problem of testing whether a given cl-graph is cl-planar.
The \clpshort problem was introduced by Forster and Bachmaier~\cite{fb-clp-04}, who showed a polynomial-time testing algorithm for the case in which the level graph is a proper hierarchy and the clusters are level-connected.

A \emph{$\mathcal{T}$-level graph} (also known as \emph{generalized $k$-ary tanglegram}) $(V,E,\gamma,\mathcal{T})$ is a level graph $(V,E,\gamma)$ equipped with a set $\mathcal{T}=T_1,\dots,T_k$ of trees such that the leaves of $T_i$ are the vertices of level $V_i$ of $(V,E,\gamma)$, for $1 \leq i \leq k$.
A \emph{$\mathcal{T}$-level planar drawing} of $(V,E,\gamma,\mathcal{T})$ is a level planar drawing of $(V,E,\gamma)$ such that, for $i = 1,\dots,k$, the order in which the vertices of $V_i$ appear along $L_i$ is \emph{compatible} with $T_i$, that is, for each node $w$ of $T_i$, the leaves of the subtree of $T_i$ rooted at $w$ appear consecutively along $L_i$.
A $\mathcal{T}$-level graph is \emph{$\mathcal{T}$-level planar} if it admits a $\mathcal{T}$-level planar drawing.
\tlp is the problem of testing whether a given $\mathcal{T}$-level graph is $\mathcal{T}$-level planar.
The \tlp problem was introduced by Wotzlaw {\em et al.}~\cite{wsp-gktlg-12}, who showed a quadratic-time algorithm for the case in which the number of vertices at each level is bounded by a constant.

The definition of {\em proper} naturally extends to cl-graphs and $\mathcal{T}$-level graphs. Note that, given any non-proper level graph $G$ it is easy to construct a proper level graph $G'$ that is level planar if and only if $G$ is level planar. However, as mentioned above, there exists no trivial transformation from a non-proper cl-graph (a non-proper $\mathcal{T}$-level graph) to an equivalent proper cl-graph (resp., an equivalent proper $\mathcal{T}$-level graph).

In this paper we show that \clp and \tlp are \NPC for non-proper instances. Conversely, we show that both problems are polynomial-time
solvable for proper instances. Our results have several consequences:
(1) They narrow the gap between polynomiality and \NPCN in the
classification of Schaefer~\cite{s-ttphtpv-13} (see
Fig.~\ref{fi:schema}). The reduction of Schaefer between \tlp and \sefe holds for proper instances~\cite{s-ttphtpv-13}. (2) They allow to partially answer a question from~\cite{s-ttphtpv-13} asking whether a reduction exists from \clpshort to \sefe. We show that such a reduction exists for proper instances and that a reduction from general instances would imply the \NPHN of \sefe. (3) They improve on~\cite{fb-clp-04} and~\cite{wsp-gktlg-12} by extending the classes of instances which are decidable in polynomial-time for \clpshort and \tlp, respectively. (4) They provide the first, as far as we know, \NPCN for a problem that has all the constraints of clustered planarity problem (and some more).

\begin{figure}[tb]
\centering
  \includegraphics[width=0.65\textwidth]{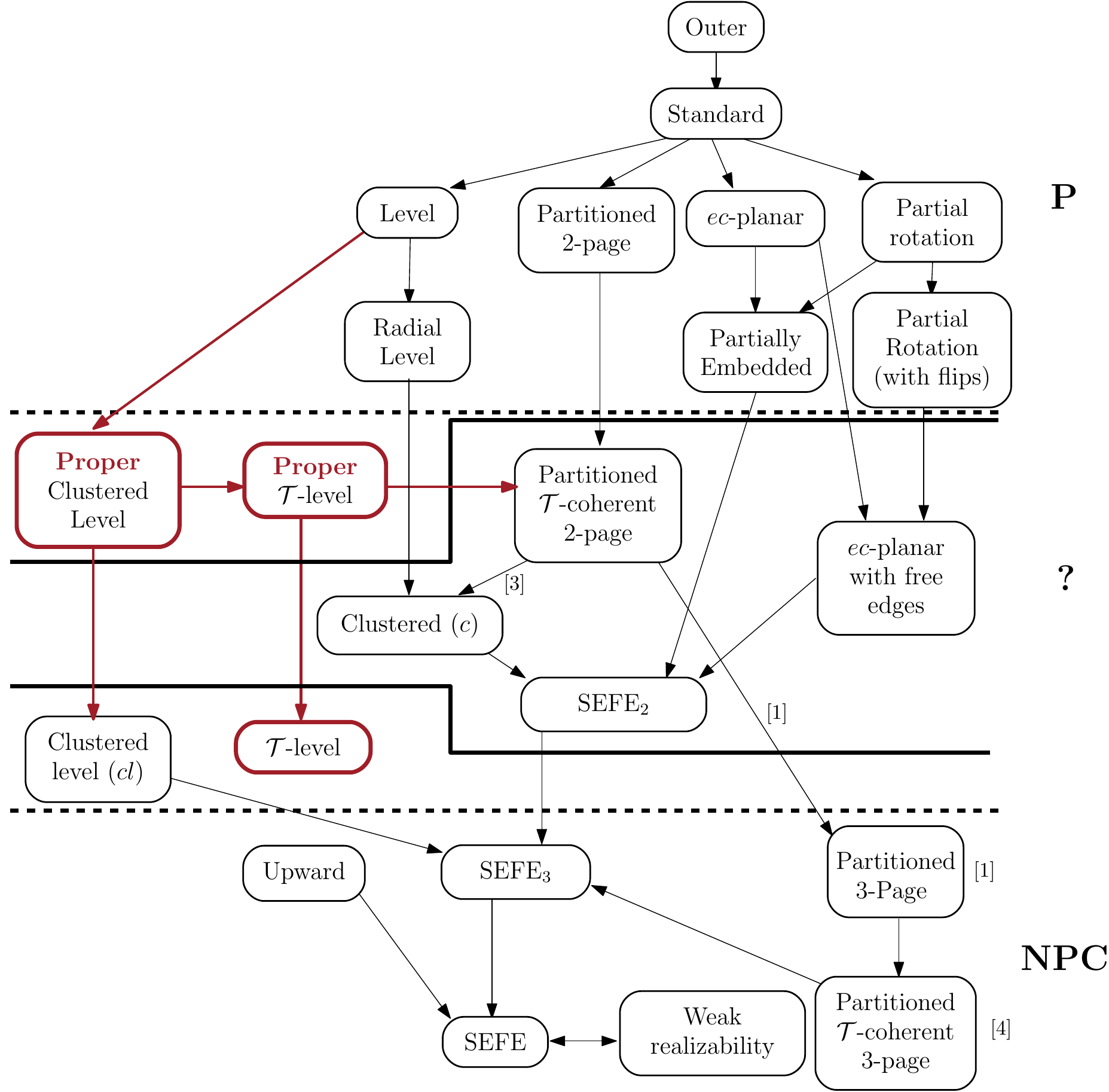}
  \caption{Updates on the classification proposed by Schaefer
    in~\cite{s-ttphtpv-13}. Dashed lines represent the boundaries
    between problems that were known to be polynomial-time solvable,
    problems that were known to be \NPC, and problems whose complexity
    was unknown before this paper. Solid lines represent the new
    boundaries according to the results of this paper. All the
    arcs representing reductions that can be transitively inferred are
    omitted. Results proved after~\cite{s-ttphtpv-13} are equipped
    with references. Reductions and classes introduced in
    this paper are red. The prefix ``proper'' has been added to two
    classes in~\cite{s-ttphtpv-13} to better clarify their
    nature. }\label{fi:schema}
\end{figure}
\nocite{adn-ocsprs-13-tr,adn-osnpsp-14,ad-drbsc-14}

The paper is organized as follows. The \NPCN proofs are in Section~\ref{se:hardness}, while the algorithms are in Section~\ref{se:poly}. We conclude with open problems in Section~\ref{se:conclusions}.

\section{NP-Hardness} \label{se:hardness}

In this section we prove that the \tlp and the \clpshort problems are
\NPC. In both cases, the \NP-hardness is proved by means of a
polynomial-time reduction from the \NPC problem \betp~\cite{top-o-79},
that takes as input a finite set $A$ of $n$ objects and a set $C$ of
$m$ ordered triples of distinct elements of $A$, and asks whether a
linear ordering $\mathcal{O}$ of the elements of $A$ exists such that
for each triple $\langle \alpha,\beta,\delta \rangle$ of $C$, we have
either $\mathcal{O} =<$$\dots, \alpha, \dots, \beta,
\dots,\delta,\ldots$$>$ or $\mathcal{O} = <$$\dots, \delta, \dots, \beta,
\dots, \alpha, \ldots$$>$.

\begin{theorem}\label{th:tl-hard}
  \tlp is \NPC.
\end{theorem}

\begin{figure}[tb]
  \centering
  \subfigure[]{\includegraphics[height=.4\textwidth,page=1]{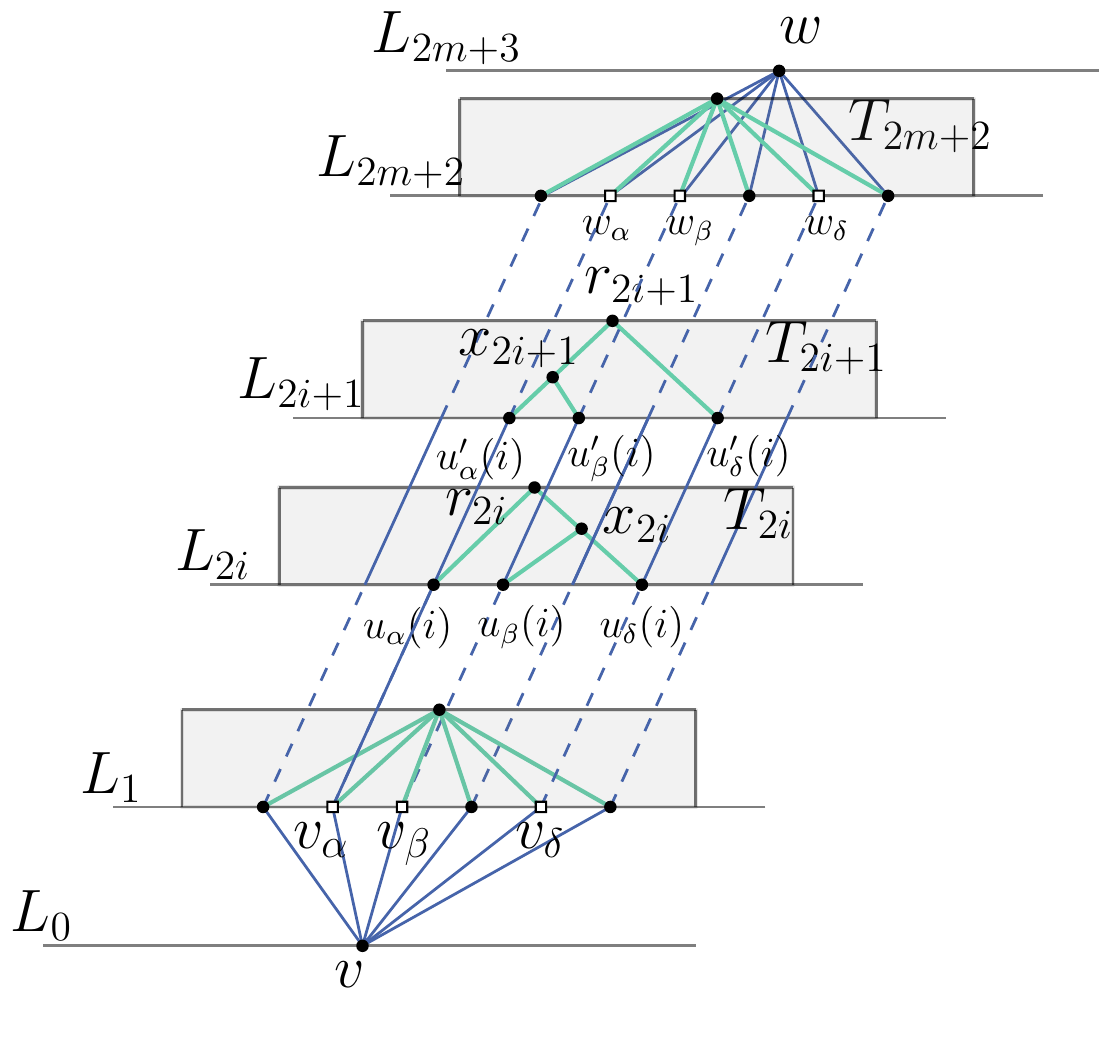}}
  \subfigure[]{\includegraphics[height=.4\textwidth,page=2]{img/hardness.pdf}}
  \caption{Illustrations for the proof of (a) Theorem~\ref{th:tl-hard} and (b) Theorem~\ref{th:cl-hard}.}\label{fig:hardness}
\end{figure}

\begin{proof}
The problem trivially belongs to \NP. We prove the \NPHN. Given an instance $\langle A, C \rangle$ of \betp, we construct an equivalent  instance \tlinstance{} of \tlp as follows. Let $A=\{1,2,\dots,n\}$ and let  $m = |C|$. Graph $(V,E)$ is composed of a set of paths connecting two vertices $v$ and $w$. Refer to Fig.~\ref{fig:hardness}(a).

Initialize $V=\{v,w\}$ and $E=\emptyset$, with $\gamma(v) = 0$ and $\gamma(w)=2m+3$. Let $T_0 \in \mathcal{T}$ and $T_{2m+3} \in \mathcal{T}$ be trees with a single node $v$ and $w$, respectively.

For each $j = 1,\dots,n$, add two vertices $v_j$ and $w_j$ to $V$, with $\gamma(v_j) = 1$ and $\gamma(w_j) = 2m+2$. Add edges $(v,v_j)$ and $(w,w_j)$ to $E$. Also, let $T_1 \in \mathcal{T}$ and $T_{2m+2} \in \mathcal{T}$ be two stars whose leaves are all the vertices of levels $V_1$ and $V_{2m+2}$, respectively.
Further, for each $j = 1,\dots,n$, we initialize variable $last(j) = v_{j}$.

Then, for each $i = 1,\dots,m$, consider the triple $t_i = \langle \alpha,\beta,\gamma \rangle$. Add six vertices $u_\alpha(i)$, $u'_\alpha(i)$, $u_\beta(i)$, $u'_\beta(i)$, $u_\delta(i)$, and $u'_\delta(i)$ to $V$ with $\gamma(u_\alpha(i)) = \gamma(u_\beta(i)) = \gamma(u_\delta(i)) = 2i$ and $\gamma(u'_\alpha(i)) = \gamma(u'_\beta(i)) = \gamma(u'_\delta(i)) = 2i+1$. Also, add edges $(last(\alpha),u_\alpha(i))$, $(last(\beta),u_\beta(i))$, $(last(\delta),u_\delta(i))$, $(u_\alpha(i), u'_\alpha(i))$, $(u_\beta(i), u'_\beta(i))$, and $(u_\gamma(i), u'_\gamma(i))$ to $E$. Further, set $last(\alpha) = u'_\alpha(i)$, $last(\beta) = u'_\beta(i)$, and $last(\delta) = u'_\delta(i)$. Let $T_{2i} \in \mathcal{T}$ be a binary tree with a root $r_{2i}$, an internal node $x_{2i}$ and a leaf $u_\alpha(i)$ both adjacent to $r_{2i}$, and with leaves $u_\beta(i)$ and $u_\delta(i)$ both adjacent to $x_{2i}$. Moreover, let $T_{2i+1}  \in \mathcal{T}$ be a binary tree with a root $r_{2i+1}$, an internal node $x_{2i+1}$ and a leaf $u'_\delta(i)$ both adjacent to $r_{2i+1}$, and with leaves $u'_\alpha(i)$ and $u'_\beta(i)$ both adjacent to $x_{2i+1}$.

Finally, for each $j = 1,\dots,n$, add an edge $(last(j),w_{j})$ to $E$.

The reduction is easily performed in $O(n+m)$ time. We prove that \tlinstance{} is $\mathcal{T}$-level planar if and only if $\langle A, C \rangle$ is a positive instance of \betp.

Suppose that \tlinstance{} admits a $\mathcal{T}$-level planar drawing $\Gamma$. Consider the left-to-right order $\mathcal{O}_1$ in which the vertices of level $V_1$ appear along $L_1$. Construct an order $\mathcal{O}$ of the elements of $A$ such that $\alpha \in A$ appears before $\beta \in A$ if and only if $v_{\alpha} \in V_1$ appears before $v_{\beta} \in V_1$ in $\mathcal{O}_1$.
In order to prove that $\mathcal{O}$ is a positive solution for $\langle A, C \rangle$, it suffices to prove that, for each triple $t_i = \langle \alpha, \beta, \delta \rangle \in C$, vertices $v_\alpha$, $v_\beta$, and $v_\delta$ appear either in this order or in the reverse order in $\mathcal{O}_1$. Note that tree $T_{2i}$ enforces $u_\alpha(i)$ not to lie between $u_\beta(i)$ and $u_\delta(i)$ along $L_{2i}$; also, tree $T_{2i+1}$ enforces $u'_\delta(i)$ not to lie between $u'_\alpha(i)$ and $u'_\beta(i)$ along $L_{2i+1}$. Since the three paths connecting $v$ and $w$ and passing through $v_\alpha$, $v_\beta$, and $v_\delta$ do not cross each other in $\Gamma$ and since they contain $u_\alpha(i)$ and $u'_\alpha(i)$, $u_\beta(i)$ and $u'_\beta(i)$, and $u_\delta(i)$ and $u'_\delta(i)$, respectively, we have that $v_\alpha$, $v_\beta$, and $v_\delta$ appear either in this order or in the reverse order in $\mathcal{O}_1$.

Suppose that an ordering $\mathcal{O}$ of the elements of $A$ exists that is a positive solution of \betp for instance $\langle A, C \rangle$.
In order to construct $\Gamma$, place the vertices of $V_1$ and $V_{2m+2}$ along $L_1$ and $L_{2m+2}$ in such a way that vertices $v_j \in V_1$ and $w_j \in V_{2m+2}$, for $j=1,\dots,n$, are assigned $x$-coordinate equal to $s$ if $j$ is the $s$-th element of $\mathcal{O}$. Also, for $i = 1, \dots, m$, let $t_i = \langle \alpha, \beta, \delta \rangle \in C$. Place vertices $u_\lambda(i)$ and $u'_\lambda(i)$, with $\lambda \in \{\alpha, \beta, \delta\}$, on $L_{2i}$ and $L_{2i+1}$, respectively, in such a way that $u_\lambda(i)$ and $u'_\lambda(i)$ are assigned  $x$-coordinate equal to $s$ if $\lambda$ is the $s$-th element of $\mathcal{O}$.
Finally, place $v$ and $w$ at any points on $L_0$ and $L_{2m+3}$, respectively, and draw the edges of $E$ as straight-line segments.
We prove that $\Gamma$ is a $\mathcal{T}$-level planar drawing of \tlinstance{}.
First note that, by construction, $\Gamma$ is a level planar drawing of $(V,E,\gamma)$. Further, for each $i = 1,\dots,m$, vertices $u_\alpha(i)$, $u_\beta(i)$, and $u_\delta(i)$ appear along $L_{2i}$ either in this order or in the reverse order; in both cases, the order is compatible with tree $T_{2i}$. Analogously, vertices $u'_\alpha(i)$, $u'_\beta(i)$, and $u'_\delta(i)$ appear along $L_{2i+1}$ either in this order or in the reverse order; in both cases, the order is compatible with tree $T_{2i+1}$. Finally, the order in which vertices of $V_0$, $V_1$, $V_{2m+2}$, and $V_{2m+3}$ appear along $L_0$, $L_1$, $L_{2m+2}$, and $L_{2m+3}$, respectively, are trivially compatible with $T_0$, $T_1$, $T_{2m+2}$, and $T_{2m+3}$.
\end{proof}

Note that the reduction described in Theorem~\ref{th:tl-hard} can be modified in such a way that $\mathcal{T}$ contains only binary trees by removing levels $V_1$ and $V_{2m+2}$. Indeed, the presence of these two levels was only meant to simplify the description of the relationship between the order of the elements of $A$ and the order of the paths between $v$ and $w$.

\begin{theorem}\label{th:cl-hard}
  \clp is \NPC.
\end{theorem}

\begin{proof}
The problem trivially belongs to class \NP. We prove the \NPHN. Given an instance $\langle A, C \rangle$ of \betp, we construct an instance \tlinstance{} of \tlp as in the proof of Theorem~\ref{th:tl-hard}; then, starting from \tlinstance{}, we construct an instance \clinstance{} of \clpshort that is cl-planar if and only if \tlinstance{} is $\mathcal{T}$-level planar. This, together with the fact that \tlinstance{} is $\mathcal{T}$-level planar if and only if $\langle A, C \rangle$ is a positive instance of \betp, implies the \NPHN of \clpshort. Refer to Fig.~\ref{fig:hardness}(b).

Cluster hierarchy $T$ is constructed as follows. Initialize $T$ with a root $\rho$. Let $w \in V_{2m+3}$ and $w_j \in V_{2m+2}$, for $j = 1,\dots,n$, be leaves of $T$ that are children of $\rho$; add an internal node $\mu_{2m+1}$ to $T$ as a child of $\rho$.
Next, for $i = m, \dots, 1$, let $u'_\delta(i)$ be a leaf of $T$ that is child of $\mu_{2i+1}$; add an internal node $\nu_{2i+1}$ to $T$ as a child of $\mu_{2i+1}$; then, let $u'_\alpha(i)$ and $u'_\beta(i)$ be leaves of $T$ that are children of $\nu_{2i+1}$; add an internal node $\mu_{2i}$ to $T$ as a child of $\nu_{2i+1}$. Further, let $u_\alpha(i)$ be a leaf of $T$ that is a child of $\mu_{2i}$; add an internal node $\nu_{2i}$ to $T$ as a child of $\mu_{2i}$; then, let $u_\beta(i)$ and $u_\delta(i)$ be leaves of $T$ that are children of $\nu_{2i}$; add an internal node $\mu_{2i-1}$ to $T$ as a child of $\nu_{2i}$.
Finally, let vertices $v \in V_0$ and $v_j \in V_{1}$, for $j = 1,\dots,n$, be leaves of $T$ that are children of $\mu_1$.

We prove that \clinstance{} is cl-planar if and only if \tlinstance{} is $\mathcal{T}$-level planar.

Suppose that \clinstance{} admits a cl-planar drawing $\Gamma$. Construct a $\mathcal{T}$-level planar drawing $\Gamma^*$ of \tlinstance{} by removing from $\Gamma$ the clusters of $T$.
First, observe that the drawing of $(V,E,\gamma)$ in $\Gamma^*$ is level-planar, since it is level-planar in $\Gamma$.
Further, for each $i = 1,\dots,m$, vertex $u_\alpha(i)$ does not appear between $u_\beta(i)$ and $u_\gamma(i)$ along line $L_{2i}$, since $u_\beta(i), u_\gamma(i) \in \nu_{2i}$ and $u_\alpha(i) \notin \nu_{2i}$; analogously, vertex $u'_\delta(i)$ does not appear between $u'_\alpha(i)$ and $u'_\beta(i)$ along line $L_{2i+1}$, since $u'_\alpha(i), u'_\beta(i) \in \nu_{2i+1}$ and $u'_\delta(i) \notin \nu_{2i+1}$. Hence, the order of the vertices of $V_{2i}$ and $V_{2i+1}$ along $L_{2i}$ and $L_{2i+1}$, respectively, are compatible with trees $T_{2i}$ and $T_{2i+1}$.
Finally, the order in which vertices of $V_0$, $V_1$, $V_{2m+2}$, and $V_{2m+3}$ appear along lines $L_0$, $L_1$, $L_{2m+2}$, and $L_{2m+3}$, respectively, are trivially compatible with $T_0$, $T_1$, $T_{2m+2}$, and $T_{2m+3}$.

Suppose that \tlinstance{} admits a $\mathcal{T}$-level planar drawing $\Gamma^*$; we describe how to construct a cl-planar drawing $\Gamma$ of \clinstance{}. Assume that $\Gamma^*$ is a straight-line drawing, which is not a loss of generality~\cite{efln-sda-06}. Initialize $\Gamma = \Gamma^*$.
Draw each cluster $\alpha$ in $T$ as a convex region $R(\alpha)$ in $\Gamma$ slightly surrounding the border of the convex hull of its vertices 
and slightly surrounding the border of the regions representing the clusters that are its descendants in $T$. Let $j$ be the largest index such 
that $V_j$ contains a vertex of $\alpha$. Then, $R(\alpha)$ contains all and only the vertices that are descendants of $\alpha$ in $T$; moreover, 
any two clusters $\alpha$ and $\beta$ in $T$ are one contained into the other, hence $R(\alpha)$ and $R(\beta)$ do not cross; finally, we prove 
that no edge $e$ in $E$ crosses more than once the boundary of $R(\alpha)$ in $\Gamma$. First, if at least one end-vertex of $e$ belongs to 
$\alpha$, then $e$ and the boundary of $R(\alpha)$ cross at most once, given that $e$ is a straight-line segment and that $R(\alpha)$ is convex. 
All the vertices in $V_0\cup \dots \cup V_{j-1}$ and at least two vertices of $V_j$ belong to $\alpha$, hence their incident edges do not cross 
the boundary of $R(\alpha)$ more than once. Further, all the vertices in $V_{j+1}\cup \dots \cup V_{2m+3}$ have $y$-coordinates larger than every 
point of $R(\alpha)$, hence edges between them do not cross $R(\alpha)$. It remains to consider the case in which $e$ connects a vertex $x_1$ in 
$V_j$ not in $\alpha$ (there is at most one such vertex) with a vertex $x_2$ in $V_{j+1}\cup \dots \cup V_{2m+2}$; in this case $e$ and 
$R(\alpha)$ do not cross given that $x_1$ is outside $R(\alpha)$, that $x_2$ has $y$-coordinate larger than every point of $R(\alpha)$, and that 
$R(\alpha)$ is arbitrarily close to the convex hull of its vertices.
\end{proof}

\section{Polynomial-Time Algorithms}\label{se:poly}

In this section we prove that problems \tlp and \clpshort become polyomial-time solvable if restricted to proper instances.

\subsection{\tlp}\label{sse:tlp}

We start by describing a polynomial-time algorithm for \tlp. The algorithm is based on a reduction to the \emph{Simultanoues Embedding with Fixed Edges} problem for two graphs (\sefe), that is defined as follows.

A \emph{simultanoues embedding with fixed edges} (SEFE) of two graphs $\red{G_1}=(V,\red{E_1})$ and $\blue{G_2}=(V,\blue{E_2})$ on the same set of vertices $V$ consists of two planar drawings $\red{\Gamma_1}$ and $\blue{\Gamma_2}$ of \Gr{} and \Gb{}, respectively, such that each vertex $v\in V$ is mapped to the same point in both drawings and each edge of the \emph{common graph} $G_\cap = (V,\red{E_1}\cap\blue{E_2})$ is represented by the same simple curve in the two drawings. The \sefe problem asks whether a given pair of graphs \sefeinstance{} admits a SEFE~\cite{bkr-sepg-12}.
The computational complexity of the \sefe problem is unknown, but there exist polynomial-time algorithms for instances that respect some conditions~\cite{adfpr-tsetgibgt-11,bkr-sepg-12,br-drpse-13,br-spqacep-13,s-ttphtpv-13}. We are going to use a result by Bl\"asius and R\"utter~\cite{br-spqacep-13}, who proposed a quadratic-time algorithm for instances \sefeinstance{} of \sefe in which \Gr{} and \Gb{} are $2$-connected, and the common graph \Gint is connected.

In the analysis of the complexity of the following algorithms we assume that the internal nodes of the trees in $\mathcal{T}$ in any instance \tlinstance{} of \tlp and of tree $T$ in any instance \clinstance{} of \clpshort have at least two children. It is easily proved that this is not a loss of generality; also, this allows us to describe the size of the instances in terms of the size of their sets of vertices.

\begin{lemma}\label{le:TCOHERENTtoSEFE}
  Let \tlinstance{} be a proper instance of \tlp. There exists an
  equivalent instance \sefeinstance{*} of \sefe such that
  $\red{G_1^*}=(V^*,\red{E^*_1})$ and $\blue{G^*_2}=(V^*,\blue{E^*_2})$ are
  $2$-connected, and the common graph $G_\cap =
  (V^*,\red{E^*_1}\cap\blue{E^*_2})$ is connected. Further, instance
  \sefeinstance{*} can be constructed in linear time.
\end{lemma}

\begin{proof}
We describe how to construct instance \sefeinstance{*}. Refer to Fig.~\ref{fig:TCOHERENTtoSEFE}.

\begin{figure}[tb]
  \centering
  \subfigure[]{\includegraphics[height=.28\textwidth,page=2]{img/k-ary.pdf}}
  \subfigure[]{\includegraphics[height=.28\textwidth,page=1]{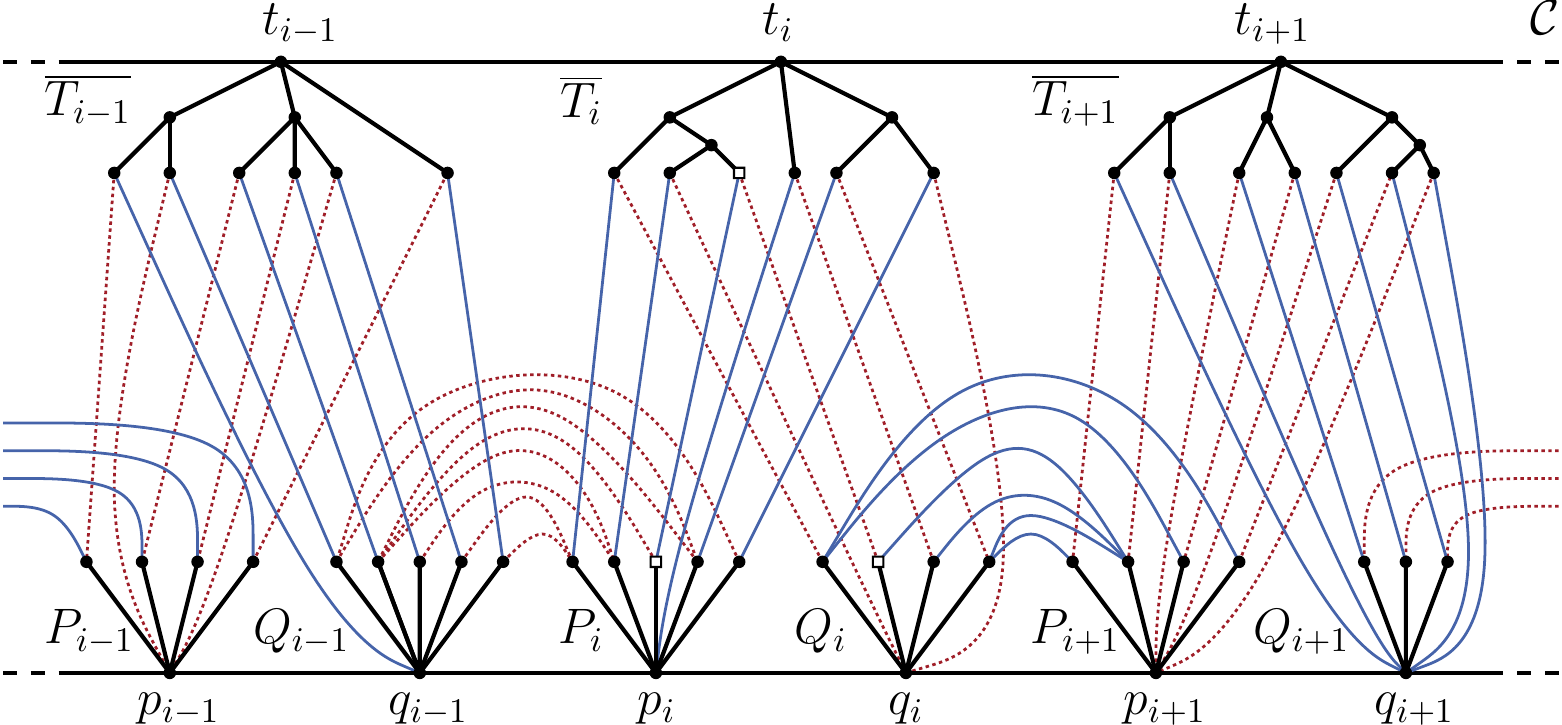}}
  \caption{Illustration for the proof of
    Lemma~\ref{le:TCOHERENTtoSEFE}. Index $i$ is assumed to be even. (a) A $T$-level planar drawing $\Gamma$ of instance \tlinstance{}. (b) The {SEFE} $\langle \red{\Gamma_1},\blue{\Gamma_2}\rangle $ of instance \sefeinstance{*} of \sefe corresponding to $\Gamma$. Correspondence between a vertex $u \in V_i$ and leaves $u(\overline{T_i}) \in \overline{T_i}$, $u(P_i) \in P_i$, and $u(Q_i) \in Q_i$ is highlighted by representing all such vertices as white boxes.}\label{fig:TCOHERENTtoSEFE}
\end{figure}

Graph \Gint contains a cycle $\mathcal{C} = t_1, t_2, \dots,$ $t_k, q_k, p_k, q_{k-1}, p_{k-1}, \dots, q_1,p_1$, where $k$ is the number of levels of \tlinstance{}. For each $i=1,\dots,k$, graph \Gint contains a copy $\overline{T_i}$ of tree $T_i \in \mathcal{T}$, whose root is identified with vertex $t_i$, and contains two stars $P_i$ and $Q_i$ centered at vertices $p_i$ and $q_i$, respectively, whose number of leaves is as follows. For each vertex $u \in V_i$ such that an edge $(u,v) \in E$ exists connecting $u$ to a vertex $v \in V_{i-1}$, star $P_i$ contains a leaf vertex $u(P_i)$; also, for each vertex $u \in V_i$ such that an edge $(u,v) \in E$ exists connecting $u$ to a vertex $v \in V_{i+1}$, star $Q_i$ contains a leaf vertex $u(Q_i)$. We also denote by $u(\overline{T_i})$ a leaf of $\overline{T_i}$ corresponding to vertex $u \in V_i$.

Graph \Gr{*} contains \Gint plus a set of edges defined as follows.
For $i = 1,\dots,k$, consider each vertex $u \in V_i$. Suppose that $i$ is even. Then, \Gr{*} has an edge connecting the leaf $u(\overline{T_i})$ of $\overline{T_i}$ corresponding to $u$ with either the leaf $u(Q_i)$ of $Q_i$ corresponding to $u$, if it exists, or with the center $q_i$ of $Q_i$, otherwise; also, for each edge in $E$ connecting a vertex $u \in V_i$ with a vertex $v \in V_{i-1}$, graph \Gr{*} has an edge connecting the leaf $u(P_i)$ of $P_i$ corresponding to $u$ with the leaf $v(Q_{i-1})$ of $Q_{i-1}$ corresponding to $v$ (such leaves exist by construction). Suppose that $i$ is odd. Then, graph \Gr{*} has an edge between $u(\overline{T_i})$ and either $u(P_i)$, if it exists, or the center $p_i$ of $P_i$, otherwise.

Graph \Gb{*} contains \Gint plus a set of edges defined as follows.
For $i = 1,\dots,k$, consider each vertex $u \in V_i$. Suppose that $i$ is odd. Then, \Gb{*} has an edge connecting $u(\overline{T_i})$ with either the leaf $u(Q_i)$ of $Q_i$ corresponding to $u$, if it exists, or with the center $q_i$ of $Q_i$, otherwise; also, for each edge in $E$ connecting a vertex $u \in V_i$ with a vertex $v \in V_{i-1}$, graph \Gb{*} has an edge $(u(P_i),v(Q_{i-1}))$. Suppose that $i$ is even. Then, graph \Gb{*} has an edge between $u(\overline{T_i})$ and either $u(P_i)$, if it exists, or $p_i$, otherwise.

It is easy to see that \Gint is connected and that \sefeinstance{*} can be constructed in polynomial time. We prove that \Gr{*} and \Gb{*} are $2$-connected, that is, removing any vertex $v$ disconnects neither \Gr{*} nor \Gb{*}. If $v$ is a leaf of either $\overline{T_i}$ or $P_i$ or $Q_i$, with $1 \leq i \leq k$, then removing $v$ disconnects neither \Gr{*} nor \Gb{*}, since \Gint remains connected.
If $v$ is an internal node (the root) of $\overline{T_i}$, or $P_i$, or $Q_i$, say of $\overline{T_i}$, with $1 \leq i \leq k$, then removing $v$ disconnects \Gint into $m = \deg(v)$ (resp. $m = \deg(v)-1$)  components, namely one component $\overline{T_i}(v)$ containing all the vertices of $\mathcal{C}$ (resp. all the vertices of $\mathcal{C}$, except for $v$) and $m-1$ subtrees $\overline{T_i}^j$ of $\overline{T_i}$, with $j=1,\dots,m-1$, rooted the children of $v$; however, by construction, each $\overline{T_i}^j$ is connected to $\overline{T_i}(v)$ via at least an edge $(u(\overline{T_i}),u(P_i)) \in \red{E^*_1}$ and an edge $(u(\overline{T_i}),u(Q_i)) \in \blue{E^*_2}$, or vice versa, incident to one of its leaves $u(\overline{T_i})$.

Observe that, if \tlinstance{} has $n_{\mathcal{T}}$ nodes in the trees of $\mathcal{T}$ (where $|V| < n_{\mathcal{T}}$), then \sefeinstance{*} contains at most $3n_{\mathcal{T}}$ vertices. Also, the number of edges of \sefeinstance{*} is at most $|E|+2n_{\mathcal{T}}$. Hence, the size of \sefeinstance{*} is linear in the size of \tlinstance{} and it is easy to see that \sefeinstance{*} can be constructed in linear time.

We prove that \sefeinstance{*} admits a SEFE if and only if \tlinstance{} is $\mathcal{T}$-level planar.

Suppose that \sefeinstance{*} admits a SEFE $\langle \red{\Gamma^*_1}, \blue{\Gamma^*_2}\rangle$. We show how to construct a drawing $\Gamma$ of \tlinstance{}. For $1 \leq i \leq k$, let $\Theta(\overline{T_i})$ be the order in which the leaves of $\overline{T_i}$ appear in a pre-order traversal of $\overline{T_i}$ in $\langle \red{\Gamma^*_1}, \blue{\Gamma^*_2}\rangle$; then, let the ordering $\mathcal{O}_i$ of the vertices of $V_i$ along $L_i$ be either $\Theta(\overline{T_i})$, if $i$ is odd, or the reverse of $\Theta(\overline{T_i})$, if $i$ is even.

We prove that $\Gamma$ is $\mathcal{T}$-level planar. For each $i = 1, \dots,k$, $\mathcal{O}_i$ is compatible with $T_i \in \mathcal{T}$, since the drawing of $\overline{T_i}$, that belongs to \Gint, is planar in $\langle \red{\Gamma^*_1}, \blue{\Gamma^*_2}\rangle$.
Suppose, for a contradiction, that two edges $(u,v), (w,z) \in E$ exist, with $u,w \in V_i$ and $v,z \in V_{i+1}$, that intersect in $\Gamma$. Hence, either $u$ appears before $w$ in $\mathcal{O}_i$ and $v$ appears after $z$ in $\mathcal{O}_{i+1}$, or vice versa. Since $i$ and $i+1$ have different parity, either $u$ appears before $w$ in $\Theta(\overline{T_i})$ and $v$ appears before $z$ in $\Theta(\overline{T_{i+1}})$, or vice versa.
We claim that, in both cases, this implies a crossing in $\langle \red{\Gamma^*_1}, \blue{\Gamma^*_2}\rangle$ between paths $(q_i,$ $u(Q_i),$ $v(P_{i+1}),$ $p_{i+1})$ and $(q_i,w(Q_i),z(P_{i+1}),p_{i+1})$ in \sefeinstance{*}. Since the edges of these two paths belong all to \Gr{*} or all to \Gb{*}, depending on whether $i$ is even or odd, this yields a contradiction.
We now prove the claim. The pre-order traversal $\Theta(Q_i)$ of $Q_i$ (the pre-order traversal $\Theta(P_{i+1})$ of $P_{i+1}$) in $\langle \red{\Gamma^*_1}, \blue{\Gamma^*_2}\rangle$ restricted to the leaves of $Q_i$ (of $P_{i+1}$) is the reverse of $\Theta(\overline{T_i})$ (of $\Theta(\overline{T_{i+1}})$) restricted to the vertices of $V_i$ (of $V_{i+1}$) corresponding to leaves of $Q_i$ (of $P_{i+1}$). Namely, each leaf $x(Q_i)$ of $Q_i$ ($y(P_{i+1})$ of $P_{i+1}$) is connected to leaf $x(\overline{T_i})$ of $\overline{T_i}$ ($y(\overline{T_{i+1}})$ of $\overline{T_{i+1}}$) in the same graph, either \Gr{*} or \Gb{*}, by construction. Hence, the fact that $u$ appears before (after) $w$ in $\Theta(\overline{T_i})$ and $v$ appears before (after) $z$ in $\Theta(\overline{T_{i+1}})$ implies that $u$ appears after (before) $w$ in $\Theta(Q_i)$ and $v$ appears after (before) $z$ in $\Theta(P_{i+1})$. In both cases, this implies a crossing in $\langle \red{\Gamma^*_1}, \blue{\Gamma^*_2}\rangle$ between the two paths.

Suppose that \tlinstance{} admits a $\mathcal{T}$-level planar drawing $\Gamma$. We show how to construct a SEFE $\langle \red{\Gamma^*_1}, \blue{\Gamma^*_2}\rangle$ of
\sefeinstance{*}.
For $1 \leq i \leq k$, consider the order $\mathcal{O}_i$ of the vertices of level $V_i$ along $L_i$ in $\Gamma$. Since $\Gamma$ is $\mathcal{T}$-level planar, there exists an embedding $\Gamma_i$ of tree $T_i \in \mathcal{T}$ that is compatible with $\mathcal{O}_i$. If $i$ is odd (even), then assign to each internal vertex of $\overline{T_i}$ the same (resp. the opposite) rotation scheme  as its corresponding vertex in $\Gamma_i$. Also, if $i$ is odd, then assign to $p_i$ (to $q_i$) the rotation scheme in \Gr{*} (in \Gb{*}) such that the paths connecting $p_i$ ($q_i$) to the leaves of $\overline{T_i}$ (either with an edge or passing through a leaf of the corresponding star of \Gint) appear in the same clockwise order as the vertices of $V_i$ appear in $\mathcal{O}_i$; if $i$ is even, then assign to $p_i$ (to $q_i$) the rotation scheme in \Gb{*} (in \Gr{*}) such that the paths connecting $p_i$ ($q_i$) to the leaves of $\overline{T_i}$ appear in the same counterclockwise order as the vertices of $V_i$ appear in $\mathcal{O}_i$.
Finally, consider the embedding $\Gamma_{i,i+1}$ obtained by restricting $\Gamma$ to the vertices and edges of the subgraph induced by the vertices of $V_i$ and $V_{i+1}$. If $i$ is odd (even), then assign to the leaves of $Q_i$ and of $P_{i+1}$ in \Gr{*} (in \Gb{*}) the same rotation scheme as their corresponding vertices have in $\Gamma_{i,i+1}$. This completes the construction of $\langle \red{\Gamma^*_1}, \blue{\Gamma^*_2}\rangle$.

We prove that $\langle \red{\Gamma^*_1}, \blue{\Gamma^*_2}\rangle$ is a SEFE of \sefeinstance{*}.
Since the rotation scheme of the internal vertices of each $\overline{T_i}$ are constructed starting from an embedding of $\Gamma_i$ of tree $T_i \in \mathcal{T}$ that is compatible with $\mathcal{O}_i$, the drawing of $\overline{T_i}$ is planar.
Further, since the rotation schemes of $p_i$ (of $q_i$) are also constructed starting from $\mathcal{O}_i$, there exists no crossing between two paths connecting $t_i$ and $p_i$ ($t_i$ and $q_i$), one passing through a leaf $u(\overline{T_i})$ of $\overline{T_i}$ and, possibly, through a leaf $u(P_i)$ of $P_i$ (through a leaf $u(Q_i)$ of $Q_i$), and the other passing through a leaf $v(\overline{T_i})$ of $\overline{T_i}$ and, possibly, through a leaf $v(P_i)$ of $P_i$ (through a leaf $v(Q_i)$ of $Q_i$).
Finally, since the rotation schemes of the leaves of $Q_i$ and $P_{i+1}$ are constructed from the embedding $\Gamma_{i,i+1}$ obtained by restricting $\Gamma$ to the vertices and edges of the subgraph induced by the vertices of $V_i$ and $V_{i+1}$, there exist no two crossing edges between leaves of $Q_i$ and of $P_{i+1}$.
\end{proof}

We remark that a reduction from \tlp to \sefe was described by Schaefer in~\cite{s-ttphtpv-13}; however, the instances of \sefe obtained from that reduction do not satisfy any conditions that make \sefe known to be solvable in polynomial-time.

\begin{theorem}\label{th:tl-polynomial}
  Let \tlinstance{} be a proper instance of \tlp. There exists a
  quadratic-time algorithm that decides whether \tlinstance{} is $\mathcal{T}$-level planar.
\end{theorem}
\begin{proof}
By Lemma~\ref{le:TCOHERENTtoSEFE}, an instance \sefeinstance{} of \sefe can be constructed in linear time such that $\red{G_1}$ and $\blue{G_2}$ are $2$-connected, the common graph $G_\cap$ is connected, and \sefeinstance{} is a positive instance of \sefe if and only if \tlinstance{} is $\mathcal{T}$-level planar. The statement follows from the fact that there exists a quadratic-time algorithm~\cite{br-spqacep-13} that decides whether \sefeinstance{} is a positive instance of \sefe.
\end{proof}

\subsection{\clp}\label{sse:clp}

In the following we prove that the polynomial-time algorithm to decide the existence of a $\mathcal{T}$-level planar drawing of a proper instance \tlinstance{} of \tlp can be also employed to decide in polynomial time the existence of a cl-planar drawing of a proper instance \clinstance{} of \clpshort.

A proper cl-graph \clinstance{} is \emph{$\mu$-connected between two levels} $V_i$ and $V_{i+1}$ if there exist two vertices $u \in V_\mu \cap V_i$ and $v \in V_\mu \cap V_{i+1}$ such that edge $(u,v) \in E$.
For a cluster $\mu \in T$, let $\gamma_{\min{}}(\mu) = \min{}\{i | V_i \cap V_\mu \neq \emptyset\}$ and let $\gamma_{\max{}}(\mu) = \max{}\{i | V_i \cap V_\mu \neq \emptyset\}$.
A proper cl-graph \clinstance{} is \emph{level-$\mu$-connected} if it is
$\mu$-connected between levels $V_i$ and $V_{i+1}$ for each $i=\gamma_{\min}(\mu), \dots,\gamma_{\max}(\mu)-1$.
A proper cl-graph \clinstance{} is \emph{level-connected} if it is
$\mu$-level-connected for each cluster $\mu \in T$.

Our strategy consists of first transforming a proper instance of \clpshort into an equivalent level-connected instance, and then transforming such a level-connected instance into an equivalent proper instance of \tlp.

\begin{lemma}\label{le:PROPERtoLEVELCONNECTED}
  Let \clinstance{} be a proper instance of \clp. There exists an equivalent
  level-connected instance \clinstance{*} of \clp.
  Further,  the size of \clinstance{*} is quadratic in the size of \clinstance{} and \clinstance{*} can be constructed in quadratic time.
\end{lemma}

\begin{proof}
The construction of \clinstance{*} works in two steps. See Fig.~\ref{fig:level-connected}.

First, we transform \clinstance{} into an equivalent instance \clinstance{\prime}. Initialize $V'=V$, $E'=E$, and $T'=T$. Also, for each $i = 1, \dots, k$ and for each vertex $u \in V_i$, set $\gamma'(u) = 3(i-1)+1$.
Then, for each $i = 1, \dots, k-1$, consider each edge $(u,v) \in E$ such that $\gamma(u)=i$ and $\gamma(v)=i+1$. Add two vertices $d_u$ and $d_v$ to $V'$, and replace $(u,v)$ in $E'$ with three edges $(u,d_u)$, $(d_u,d_v)$, and $(d_v,v)$. Set $\gamma'(d_u) = 3(i-1)+2$ and $\gamma'(d_v) = 3i$. Finally, add $d_u$ ($d_v$) to $T'$ as a child of the parent of $u$ (of $v$) in $T'$.

Second, we transform \clinstance{\prime} into an equivalent level-connected instance \clinstance{*}. Initialize $V^*=V'$, $E^*=E'$, $\gamma^*=\gamma'$, and $T^*=T'$. Consider each cluster $\mu \in T'$ according to a bottom-up visit of $T'$. If there exists a level $V'_i$, with $\gamma'_{\min{}}(\mu) \leq i < \gamma'_{\max{}}(\mu)$, such that no edge in $E'$ connects a vertex $u \in V'_i \cap V'_\mu$ with a vertex $v \in V'_{i+1} \cap V'_\mu$, then add two vertices $u^*$ and $v^*$ to $V^*$, add an edge $(u^*,v^*)$ to $E^*$, set $\gamma^*(u^*) = i$ and $\gamma^*(v^*) = i+1$, and add $u^*$ and $v^*$ to $T^*$ as children of $\mu$.

\begin{figure}[tb]
  \centering
  \subfigure[]{\includegraphics[width=.15\textwidth,page=1]{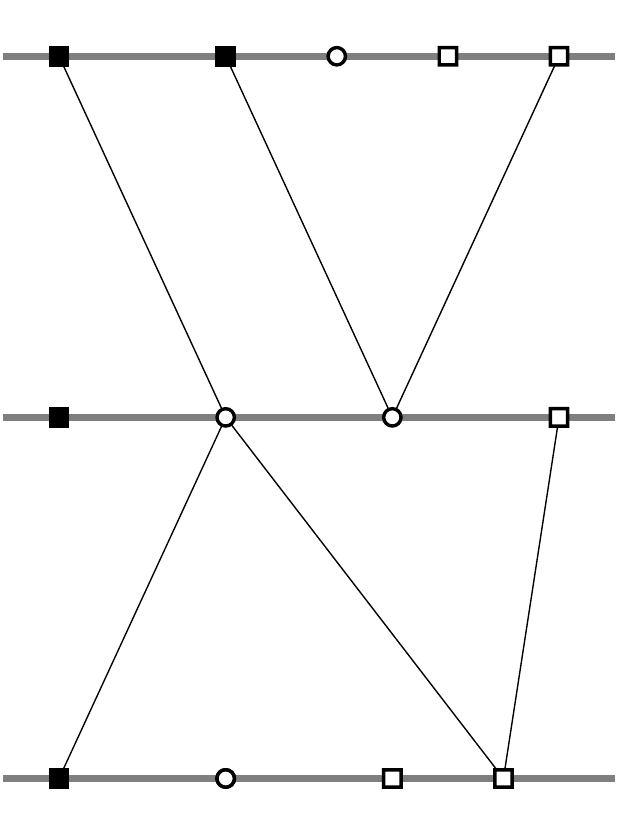}}
  \hspace{1cm}
  \subfigure[]{\includegraphics[width=.15\textwidth,page=1]{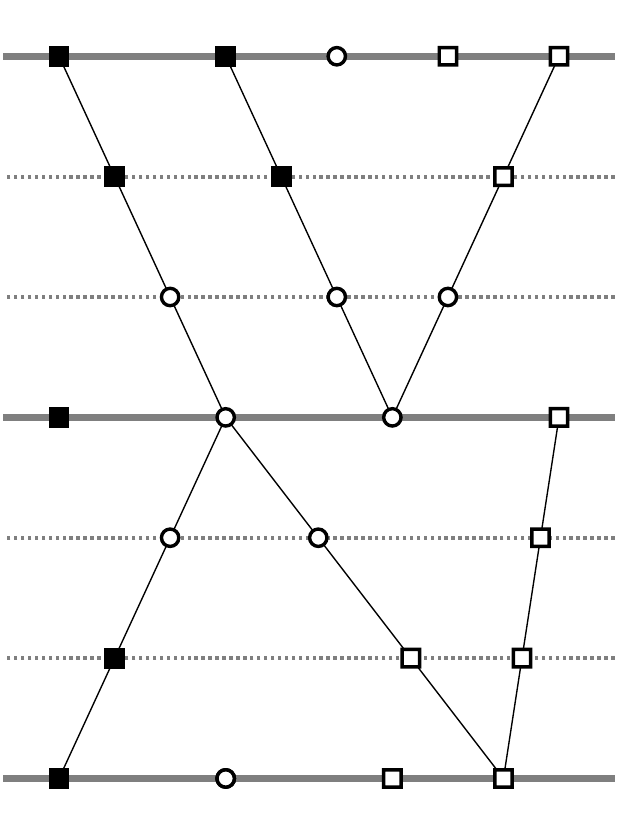}}
  \hspace{1cm}
  \subfigure[]{\includegraphics[width=.15\textwidth,page=1]{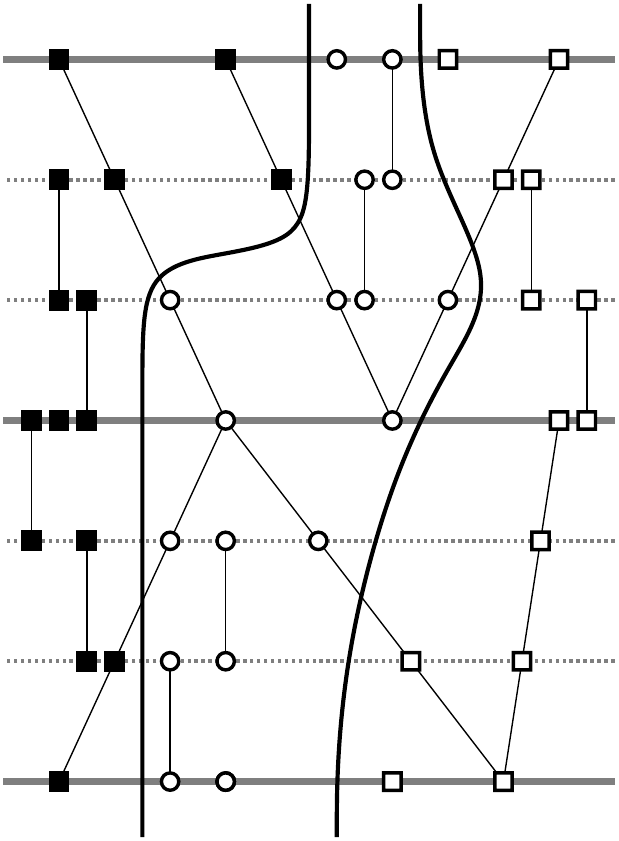}}
\caption{Illustration for the proof of Lemma~\ref{le:PROPERtoLEVELCONNECTED}. (a) Instance \clinstance{} with flat hierarchy containing clusters $\mu_\blacksquare$, $\mu_\square$, and $\mu_\circ$. (b) Insertion of dummy vertices in \clinstance{} to obtain \clinstance{\prime}. (c) Level-connected instance \clinstance{*} obtained from \clinstance{\prime}.}\label{fig:level-connected}
\end{figure}

Observe that, for each cluster $\mu \in T$ and for each level $1 \leq i \leq 3k-2$, at most two dummy vertices are added to \clinstance{*}. This implies that $|V^*| \in O(|V|^2)$. Also, the whole construction can be performed in $O(|V|^2)$ time.

\begin{claimx}\label{cl:one}
\clinstance{\prime} is equivalent to \clinstance{}.
\end{claimx}
\begin{proof}
Suppose that \clinstance{} admits a cl-planar drawing $\Gamma$; we show how to construct a cl-planar drawing $\Gamma'$ of \clinstance{\prime}. 
Initialize $\Gamma'=\Gamma$. We scale $\Gamma'$ up by a factor of $3$ and we vertically translate it so that the vertices in $V'_1$ lie on line 
$y=1$. After the two affine transformations have been applied (i) no crossing has been introduced in the drawing, (ii) every edge is still drawn 
as a $y$-monotone curve, (iii) for  $i=1,\dots,k$, the vertices of level $V_i=V'_{3(i-1)+1}$ are placed on line $y=3(i-1)+1$, that we denote by 
$L'_{3(i-1)+1}$, and (iv) the order in which vertices of $V_i=V'_{3(i-1)+1}$ appear along $L'_{3(i-1)+1}$ is the same as the order in which they 
appeared along $L_i$. For each $i = 1, \dots, k-1$, consider each edge $(u,v) \in E$ such that $\gamma(u)=i$ and $\gamma(v)=i+1$. Place vertices 
$d_u$ and $d_v$ in $\Gamma'$ on the two points of the curve representing $(u,v)$ having $y$-coordinate $3(i-1)+2$ and $3i$, respectively. Then, 
the curves representing in $\Gamma'$ any two edges in $E'$ are part of the curves representing in $\Gamma'$ any two edges in $E$. Hence $\Gamma'$ 
is a cl-planar drawing of \clinstance{\prime}.

Suppose that \clinstance{\prime} admits a cl-planar drawing $\Gamma'$; we show how to construct a cl-planar drawing $\Gamma$ of \clinstance{}. 
Initialize $\Gamma=\Gamma'$. For each $i=1,\dots,k-1$, consider each path $(u,d_u,d_v,v)$ such that $\gamma'(u)=3(i-1)+1$, 
$\gamma'(d_u)=3(i-1)+2$, $\gamma'(d_v)=3i$, and $\gamma'(v)=3i+1$; remove vertices $d_u$ and $d_v$, and their incident edges in $E'$ from 
$\Gamma$; draw edge $(u,v) \in E$ as a curve obtained as a composition of the curves representing edges $(u,d_u)$, $(d_u,d_v)$, and $(d_v,v)$ in 
$\Gamma'$. Scale $\Gamma$ down by a factor of $3$ and vertically translate it so that the vertices of $V_1$ lie on line $y=1$. After the two 
affine transformations have been applied (i) no crossing has been introduced in the drawing, (ii) every edge is still drawn as a $y$-monotone 
curve, (iii) for $i=1,\dots,k$, the vertices of level $V_i$ are placed on line $y=i$, and (iv) the order in which vertices of $V_i=V'_{3(i-1)+1}$ 
appear along $L_i$ is the same as the order in which they appeared along $L_{3(i-1)+1}'$. Since $\Gamma'$ is cl-planar, this implies that 
$\Gamma$ is cl-planar, as well.
\end{proof}

Instance \clinstance{\prime} is such that, if there exists a vertex $u \in V'_{j}$, with $1 \leq j \leq 3(k-1)+1$, that is adjacent to two vertices $v, w \in V'_{h}$, with $h = j \pm 1$, then $u$, $v$, and $w$ have the same parent node $\mu \in T'$; hence, \clinstance{\prime} is $\mu$-connected between levels $V'_j$ and $V'_{h}$.

\begin{claimx}\label{cl:two}
\clinstance{*} is equivalent to \clinstance{\prime}.
\end{claimx}
\begin{proof}
Suppose that \clinstance{*} admits a cl-planar drawing $\Gamma^*$; we show how to construct a cl-planar drawing $\Gamma'$ of \clinstance{\prime}. Initialize $\Gamma'=\Gamma^*$ and remove from $V'$, $E'$, and $\Gamma'$ all the vertices and edges added when constructing $\Gamma^*$. Since all the other vertices of $V'$ and edges of $E'$ have the same representation in $\Gamma'$ and in $\Gamma^*$, and since $\Gamma^*$ is cl-planar, drawing $\Gamma'$ is cl-planar, as well.

Suppose that \clinstance{\prime} admits a cl-planar drawing $\Gamma'$; we show how to construct a cl-planar drawing $\Gamma^*$ of \clinstance{*}. Initialize $\Gamma^*=\Gamma'$. Consider a level $V'_i$, with $1 \leq i \leq 3(k-1)$, such that vertices $u^*, v^* \in \mu$ with $\gamma'(u^*) = i$ and $\gamma'(v^*) = i+1$, for some cluster $\mu \in T$, have been added to \clinstance{*}. By construction, \clinstance{\prime} is not $\mu$-connected between levels $V'_i$ and $V'_{i+1}$. As observed before, this implies that no vertex $u \in V'_{i} \cap V'_\mu$ exists that is connected to two vertices $v, w \in V'_{i+1}$, and no vertex $u \in V'_{i+1} \cap V'_\mu$ exists that is connected to two vertices $v, w \in V'_{i}$. Hence, vertices $u^*$ and $v^*$, and edge $(u^*,v^*)$, can be drawn in $\Gamma^*$ entirely inside the region representing $\mu$ in such a way that $u^*$ and $v^*$ lie along lines $L'_i$ and $L'_{i+1}$ and there exists no crossing between edge $(u^*,v^*)$ and another edge.
\end{proof}

This concludes the proof of the lemma.
\end{proof}

\begin{lemma}\label{le:CLPtoTL}
  Let \clinstance{} be a level-connected instance of \clp. There exists an equivalent proper instance \tlinstance{} of \tlp.  Further, the size of \tlinstance{} is linear in the size of \clinstance{} and \tlinstance{} can be constructed in quadratic time.
\end{lemma}

\begin{proof}
We construct \tlinstance{} from \clinstance{} as follows. Initialize $\mathcal{T} = \emptyset$. For $i=1,\dots,k$, add to $\mathcal{T}$ a tree $T_i$ that is the subtree of the cluster hierarchy $T$ whose leaves are all and only the vertices of level $V_i$. Note that the set of leaves of the trees in $\mathcal{T}$ corresponds to the vertex set $V$. Since each internal node of the trees in $\mathcal{T}$ has at least two children, we have that the size of \tlinstance{} is linear in the size of \clinstance{}. Also, the construction of \tlinstance{} can be easily performed in $O(|V|^2)$ time.

We prove that \tlinstance{} is $\mathcal{T}$-level planar if and only if \clinstance{} is cl-planar.

Suppose that \tlinstance{} admits a $\mathcal{T}$-level planar drawing $\Gamma^*$; we show how to construct a cl-planar drawing $\Gamma$ of \clinstance{}. Initialize $\Gamma = \Gamma^*$. Consider each level $V_i$, with $i = 1,\dots,k$. By construction, for each cluster $\mu \in T$ such that there exists a vertex $v \in V_i \cap V_\mu$, there exists an internal node of tree $T_i \in \mathcal{T}$ whose leaves are all and only the vertices of $V_i \cap V_\mu$. Since $\Gamma^*$ is $\mathcal{T}$-level planar, such vertices appear consecutively along $L_i$. Hence, in order to prove that $\Gamma$ is a cl-planar drawing, it suffices to prove that there exist no four vertices $u,v,w,z$ such that (i) $u,v \in V_i$ and $w,z \in V_j$,  with $1 \leq i < j \leq k$; (ii) $u,w \in V_\mu$ and $v,z \in V_\nu$, with $\mu \neq \nu$; and (iii) $u$ appears before $v$ on $L_i$ and $w$ appears after $z$ on $L_j$, or vice versa.
Suppose, for a contradiction, that such four vertices exist. Note that, we can assume $j = i \pm 1$ without loss of generality, as \clinstance{} is level-connected. Assume that $u$ appears before $v$ along $L_i$ and $w$ appears after $z$ along $L_j$, the other case being symmetric. Since $\Gamma^*$ is $\mathcal{T}$-level planar, all the vertices of $V_\mu$ appear before all the vertices of $V_\nu$ along $L_i$ and all the vertices of $V_\mu$ appear after all the vertices of $V_\nu$ along $L_j$. Also, since \clinstance{} is level-connected, there exists at least an edge $(a,b)$ such that $a \in V_i \cap V_\mu$ and $b \in V_j \cap V_\mu$, and an edge $(c,d)$ such that $c \in V_i \cap V_\nu$ and $d \in V_j \cap V_\nu$. However, under the above conditions, these two edges intersect in $\Gamma$ and in $\Gamma^*$, hence contradicting the hypothesis that $\Gamma^*$ is $\mathcal{T}$-level planar.

Suppose that \clinstance{} admits a cl-planar drawing $\Gamma$; we show how to construct a $\mathcal{T}$-level planar drawing $\Gamma^*$ of \tlinstance{}. Initialize $\Gamma^* = \Gamma$. Consider each level $V_i$, with $i = 1,\dots,k$. By construction, for each internal node $w$ of tree $T_i \in \mathcal{T}$, there exists a cluster $\mu \in T$ such that the vertices of $V_i \cap V_\mu$ are all and only the leaves of the subtree of $T_i$ rooted at $w$. Since $\Gamma$ is cl-planar, such vertices appear consecutively along $L_i$. Hence, $\Gamma^*$ is $\mathcal{T}$-level planar.
\end{proof}

We get the following.

\begin{theorem}\label{th:cl-plynomial}
  Let \clinstance{} be a proper instance of \clp. There exists an   $O(|V|^4)$-time algorithm that decides whether \clinstance{} admits a cl-planar drawing.
\end{theorem}
\begin{proof}
By Lemma~\ref{le:PROPERtoLEVELCONNECTED}, it is possible to construct in $O(|V|^2)$ time a level-connected instance \clinstance{\prime} of \clpshort that is cl-planar if and only if \clinstance{} is cl-planar, with $|V'|=O(|V|^2)$. By Lemma~\ref{le:CLPtoTL}, it is possible to construct in $O(|V'|^2)$ time a proper instance \tlinstance{\prime} of \tlp that is $\mathcal{T}$-level planar if and only if \clinstance{\prime} is cl-planar. Finally, by Theorem~\ref{th:tl-polynomial}, it is possible to test in $O(|V'|^2)$ time whether \tlinstance{\prime} is $\mathcal{T}$-level planar.
\end{proof}

\clearpage
\section{Open Problems}\label{se:conclusions}

Several problems are opened by this research:
\begin{enumerate}
\item The algorithm in~\cite{jl-lpelt-02} for testing level planarity and the algorithm in~\cite{fb-clp-04} for testing \clpshort for level-connected instances in which the level graph is a proper hierarchy both have linear-time complexity. The algorithm in~\cite{wsp-gktlg-12} for testing \tlp for instances in which the number of vertices on each level is bounded by a constant has quadratic-time complexity. Although our polynomial-time algorithms solve more general problems than the ones cited above, they are less efficient. Hence, there is room for future research aiming at improving our complexity bounds.
\item Our \NPCN result on the complexity of \clpshort exploits a cluster hierarchy whose depth is linear in the number of vertices of the underlying graph. Does the \NPHN hold even when the hierarchy is flat or has a depth that is sublinear in the number of vertices?
\item The \NPHN of \clpshort presented in this paper is, to the best of our knowledge, the first hardness result for a variation of the clustered planarity problem in which none of the c-planarity constraints is dropped. Is it possible to use similar techniques to tackle the more intriguing problem of determining the complexity of {\sc Clustered Planarity}?
\end{enumerate}

{\small \bibliography{bibliography}} \bibliographystyle{splncs03}
\end{document}